\documentclass[letterpaper, 10 pt, conference]{ieeeconf}  

\IEEEoverridecommandlockouts                              
\title{\LARGE \bf
A Frequency-Domain Stability Method for Reset Systems
}

\author{Ali Ahmadi Dastjerdi, Alessandro Astolfi, and S. Hassan HosseinNia
\thanks{A. Astolfi is with the Department of Electrical and Electronic Engineering, Imperial College London, London, SW7 2AZ, UK and with the Dipartimento di Ingegneria Civile e Ingegneria Informatica, Universita di Roma ``Tor Vergata", Rome, 00133, Italy
        {\tt\small a.astolfi@imperial.ac.uk}}%
\thanks{S.H. HosseinNia and A. Ahmadi Dastjerdi are with the Faculty of Precision and Microsystem  Engineering, Delft University of Technology, 2826 CD Delft, The Netherlands
        {\tt\small S.H.HosseinNiaKani@tudelft.nl}
        {\tt\small A.AhmadiDastjerdi@tudelft.nl}}%
}

\usepackage{graphics} 
\usepackage{epsfig} 
\usepackage{mathptmx} 
\usepackage{times} 
\usepackage{amsmath} 
\usepackage{amssymb}
\usepackage{amsfonts}
\usepackage{amsmath}
\usepackage{amssymb}
\usepackage{mathtools}
\usepackage{steinmetz}
\usepackage{cite}
\newtheorem{definition}{Definition}
\newtheorem{theorem}{Theorem}
\newtheorem{remark}{Remark}

\usepackage{enumerate}
\usepackage{esvect}
\usepackage{cancel}
\usepackage{tikz} 
\usetikzlibrary{calc,patterns,arrows,shapes.arrows,intersections}
\usepackage{mathtools}
\DeclarePairedDelimiter\abs{\lvert}{\rvert}%
\DeclarePairedDelimiter\norm{\lVert}{\rVert}%
\makeatletter
\let\oldabs\abs
\def\abs{\@ifstar{\oldabs}{\oldabs*}}
\let\oldnorm\norm
\def\norm{\@ifstar{\oldnorm}{\oldnorm*}}
\makeatother
\DeclareMathAlphabet{\mathcal}{OMS}{cmsy}{m}{n}
\begin{document}

\maketitle
\thispagestyle{empty}
\pagestyle{empty}

\begin{abstract}

Nowadays, the demand for an alternative to linear PID controllers has increased because of the rising expectations of the high-precision industry. The potential of reset controllers to solve this important challenge has been extensively demonstrated in the literature. However, similarly to other non-linear controllers, the stability analysis for these controllers is complex and relies on parametric models of the systems which may hinder the applicability of these controllers in industry. The well-known $H_\beta$ method tries to solve this significant issue. However, assessing the $H_\beta$ condition in the frequency-domain is complex, especially for high dimensional plants. In addition, it cannot be used to assess UBIBS stability of reset control systems in the case of reseting to non-zero values. These problems have been solved in this paper for the first order reset elements, and an easy-to-use frequency approach for assessing stability of reset control systems is proposed. The effectiveness of the proposed approach is demonstrated through a practical example. 

\end{abstract}

\section{INTRODUCTION}\label{sec:1}

Technology developments in cutting edge industries have control requirements that cannot be fulfilled by linear controllers. To overcome this problem, linear controllers should be substituted with non-linear ones, for example reset controllers. These controllers have attracted significant attention due to their simple structure \cite{clegg1958nonlinear,beker2004fundamental,aangenent2010performance,forni2011reset,villaverde2011reset,banos2011reset,van2017frequency,hosseinnia2013fractional,guo2015analysis}. The advantages of reset controllers have been utilized to enhance the performance of several mechatronic systems (see, e.g. \cite{horowitz1975non,hazeleger2016second,guo2009frequency,van2018hybrid,chen2019development,valerio2019reset,saikumar2019constant}). In 1958, the first reset element was introduced by Clegg \cite{clegg1958nonlinear}. The Clegg Integrator (CI) is an integrator which resets its state to zero when its input signal crosses zero. Extensions of the CI, which provide additional design freedom and flexibility, include First Order Reset Elements (FORE) \cite{zaccarian2005first,horowitz1975non}, Generalized First Order Reset Element (GFORE) \cite{saikumar2019constant}, Second Order Reset Elements (SORE) \cite {hazeleger2016second}, and Generalized Second Order Reset Element (GSORE) \cite{saikumar2019constant}. Several reset techniques, such as those based on reset bands \cite{barreiro2014reset,banos2014tuning}, fixed reset instants, partial reset (resetting to a non-zero value or resetting a selection of the controller states) \cite{zheng2007improved}, and the PI+CI approach \cite{zheng2007improved} have also been studied to improve the performances of these controllers.  

Stability is one of the most important requirements of every control system, and reset control systems are no exception \cite{khalil2002nonlinear,beker2004fundamental,van2017frequency,guo2015analysis,banos2011reset,onevsic2008stability,banos2010reset,rifai2006compositional}. Several researchers have analyzed the stability of reset controllers using quadratic Lyapunov functions \cite{banos2011reset,guo2015analysis,polenkova2012stability,vettori2014geometric}, reset instants dependant methods \cite{banos2010reset,banos2007reset,paesa2011design}, passivity, small gain, and IQC approaches \cite{khalil2002nonlinear,griggs2007stability,carrasco2010passivity,hollot1997stability}. However, most of these approaches are complex, need parametric models of the system, require solving LMI's, and are only applicable to specific types of plants. As a result, these methods do not interface well with the current control design in industry which favours the use of frequency-domain methods. Several researchers have proposed frequency-domain approaches for assessing stability of reset controllers \cite{beker1999stability,beker2004fundamental,van2017frequency}. In \cite{beker1999stability}, an approach for determining stability of a FORE in closed-loop with a mass-spring damper system has been proposed. The result in \cite{van2017frequency} is applicable to reset control systems under the specific condition $e(t)u(t)<\dfrac{u^2}{\varepsilon},\ \varepsilon>0$, in which $e(t)$ and $u(t)$ are the input and the output of the reset controller, respectively. This method is not usable in the case of partial reset techniques. 

The $H_\beta$ condition has gained significant attention among existing approaches for assessing stability of reset systems \cite{beker2004fundamental,banos2010reset,guo2015analysis}. When the base linear system of the reset controller is a first order transfer function, it provides sufficient frequency-domain conditions for uniform bounded-input bounded-state (UBIBS) stability. However, assessing the $H_\beta$ condition in the frequency-domain is complex, especially for high dimensional plants. Moreover, it cannot be used to assess UBIBS stability of reset control systems in the case of partial reset techniques. As a result, obtaining a general easy-to-use frequency-domain method for assessing stability of reset control systems is an important open problem. 

In this paper, based on the $H_\beta$ condition, a novel frequency-domain method for reset controllers with first order base linear system is proposed. This can assess UBIBS stability of reset control systems in the frequency-domain. In this method, stability is determined on the basis of the frequency response of the base linear open-loop transfer function, and the $H_\beta$ condition does not have to be explicitly tested. Besides, this method is applicable to partial reset techniques. 

The remainder of the paper is organized as follows. In Section \ref{sec:2} the problem is formulated. In Section \ref{sec:3} the frequency-domain approach for determining stability of reset control systems is presented. In Section \ref{sec:4} the effectiveness of this approach is demonstrated via a practical example. Finally, some remarks and suggestions for future studies are presented in Section \ref{sec:5}.
 \section{Problem formulation}\label{sec:2}
 In this section the well-known reset structures GFORE and  Proportional Clegg Integrator (PCI) are recalled. Then, the problem under investigation is posed. The focus of the paper is on the single-input-single-output (SISO) control architecture illustrated in Fig. \ref{F-21}. The closed-loop system consists of a linear plant with transfer function $G(s)$, a linear controller with transfer function $C_L(s)$, and a reset controller with base linear transfer function $C_R(s)$. 
 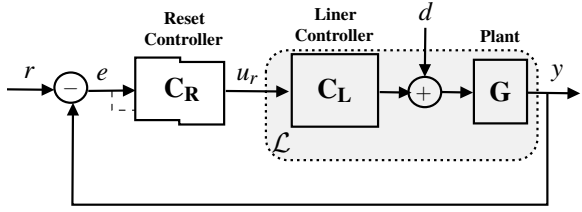
\begin{figure}[!t]
 	\centering
 	\resizebox{0.9\columnwidth}{!}{%
 		\tikzset{every picture/.style={line width=0.75pt}} 
 		\begin{tikzpicture}[x=0.75pt,y=0.75pt,yscale=-1,xscale=1]
 		\draw  [fill={rgb, 255:red, 241; green, 241; blue, 241 }  ,fill opacity=1 ][dash pattern={on 1.69pt off 2.76pt}][line width=1.5]  (238,72.45) .. controls (238,61.4) and (246.95,52.45) .. (258,52.45) -- (465,52.45) .. controls (476.05,52.45) and (485,61.4) .. (485,72.45) -- (485,132.45) .. controls (485,143.5) and (476.05,152.45) .. (465,152.45) -- (258,152.45) .. controls (246.95,152.45) and (238,143.5) .. (238,132.45) -- cycle ;
 		\draw  [line width=1.5]  (261.5,56) -- (340,56) -- (340,122) -- (261.5,122) -- cycle ;
 		\draw  [line width=1.5]  (45.63,86.45) .. controls (45.63,78.84) and (52.3,72.68) .. (60.53,72.68) .. controls (68.76,72.68) and (75.43,78.84) .. (75.43,86.45) .. controls (75.43,94.06) and (68.76,100.22) .. (60.53,100.22) .. controls (52.3,100.22) and (45.63,94.06) .. (45.63,86.45) -- cycle ;
 		\draw [line width=1.5]    (494,92) -- (494,193) -- (66.01,193.01) -- (62,193) -- (62.51,104.22) ;
 		\draw [shift={(62.53,100.22)}, rotate = 450.33] [fill={rgb, 255:red, 0; green, 0; blue, 0 }  ][line width=0.08]  [draw opacity=0] (11.61,-5.58) -- (0,0) -- (11.61,5.58) -- cycle    ;
 		\draw [line width=1.5]    (2.5,88) -- (41.63,88.41) ;
 		\draw [shift={(45.63,88.45)}, rotate = 180.6] [fill={rgb, 255:red, 0; green, 0; blue, 0 }  ][line width=0.08]  [draw opacity=0] (11.61,-5.58) -- (0,0) -- (11.61,5.58) -- cycle    ;
 		\draw [line width=1.5]    (475.65,91.05) -- (519.88,91.63) ;
 		\draw [shift={(523.88,91.68)}, rotate = 180.75] [fill={rgb, 255:red, 0; green, 0; blue, 0 }  ][line width=0.08]  [draw opacity=0] (11.61,-5.58) -- (0,0) -- (11.61,5.58) -- cycle    ;
 		\draw [line width=1.5]    (341.49,90.5) -- (366.5,90.93) ;
 		\draw [shift={(370.5,91)}, rotate = 180.99] [fill={rgb, 255:red, 0; green, 0; blue, 0 }  ][line width=0.08]  [draw opacity=0] (11.61,-5.58) -- (0,0) -- (11.61,5.58) -- cycle    ;
 		\draw  [line width=1.5]  (427.17,61.09) -- (475.5,61.09) -- (475.5,119) -- (427.17,119) -- cycle ;
 		\draw [line width=1.5]    (76.43,89.45) -- (115.5,89.04) ;
 		\draw [shift={(119.5,89)}, rotate = 539.4] [fill={rgb, 255:red, 0; green, 0; blue, 0 }  ][line width=0.08]  [draw opacity=0] (11.61,-5.58) -- (0,0) -- (11.61,5.58) -- cycle    ;
 		\draw [line width=1.5]    (201.49,89.5) -- (255,89.97) ;
 		\draw [shift={(259,90)}, rotate = 180.5] [fill={rgb, 255:red, 0; green, 0; blue, 0 }  ][line width=0.08]  [draw opacity=0] (11.61,-5.58) -- (0,0) -- (11.61,5.58) -- cycle    ;
 		\draw  [line width=1.5]  (367.63,90.45) .. controls (367.63,82.84) and (374.3,76.68) .. (382.53,76.68) .. controls (390.76,76.68) and (397.43,82.84) .. (397.43,90.45) .. controls (397.43,98.06) and (390.76,104.22) .. (382.53,104.22) .. controls (374.3,104.22) and (367.63,98.06) .. (367.63,90.45) -- cycle ;
 		\draw [line width=1.5]    (383,32) -- (382.57,72.68) ;
 		\draw [shift={(382.53,76.68)}, rotate = 270.6] [fill={rgb, 255:red, 0; green, 0; blue, 0 }  ][line width=0.08]  [draw opacity=0] (11.61,-5.58) -- (0,0) -- (11.61,5.58) -- cycle    ;
 		\draw  [line width=1.5]  (160,59.09) -- (160,63.09) -- (200,63.09) -- (200,120) -- (159.17,120) -- (159.17,116) -- (119.17,116) -- (119.17,59.09) -- (160,59.09) -- cycle ;
 		\draw [line width=1.5]    (397.43,90.45) -- (422.44,90.88) ;
 		\draw [shift={(426.44,90.95)}, rotate = 180.99] [fill={rgb, 255:red, 0; green, 0; blue, 0 }  ][line width=0.08]  [draw opacity=0] (11.61,-5.58) -- (0,0) -- (11.61,5.58) -- cycle    ;
 		\draw  [dash pattern={on 4.5pt off 4.5pt}]  (97.97,89.23) -- (98,107) -- (119,107) ;
 		
 		\draw (60.53,87.61) node  [scale=1.2,font=\large]  {$-$};
 		\draw (22.84,72.76) node   [scale=1.5,font=\large]  {$r$};
 		\draw (503,73) node  [scale=1.5,font=\large]  {$y$};
 		\draw (89.35,74.76) node   [scale=1.5,font=\large]  {$e$};
 		\draw (300,29) node  [font=\large] [align=left] {{\fontfamily{ptm}\selectfont {\large \textbf{ \ \ Liner }}}\\{\fontfamily{ptm}\selectfont {\large \textbf{Controller}}}};
 		\draw (451,40) node  [font=\large] [align=left] {{\fontfamily{ptm}\selectfont \textbf{{\large Plant}}}};
 		\draw (381.53,91.61) node  [scale=1.2,font=\large]  {$+$};
 		\draw (383.84,17) node  [scale=1.5,font=\large]  {$d$};
 		\draw (162,88) node  [scale=1.6,font=\large]  {$\mathbf{C_{\mathbf{R}}}$};
 		\draw (163,33) node  [font=\large] [align=left] {{\fontfamily{ptm}\selectfont {\large \textbf{ \ \ Reset }}}\\{\fontfamily{ptm}\selectfont {\large \textbf{Controller}}}};
 		\draw (300.75,89) node  [scale=1.6,font=\large]  {$\mathbf{C_{\mathbf{L}}}$};
 		\draw (221.35,76) node  [scale=1.5,font=\large]  {$u_{r}$};
 		\draw (451.34,90.04) node  [scale=1.6,font=\large]  {$\mathbf{G}$};
 		\draw (252,136) node  [scale=1.6,font=\large]  {$\mathcal{L}$};
 		\end{tikzpicture}
 	}
 	\caption{The closed-loop architecture of a reset controller}
 	\label{F-21}
 \end{figure}
The state-space representation of the first order reset controller is
 \begin{equation}\label{E-21}
 \begin{cases} 
 \dot{x}_r(t)=A_rx_r(t)+B_re(t), & e(t)\neq0,  \\
 x_r(t^+)=\gamma x(t), & e(t)=0, \\
 u_r(t)=C_rx(t)+D_re(t),
 \end{cases}
 \end{equation}
 in which $x_r(t)\in\mathbb{R}$ is the reset state, $A_r$, $B_r$, and $C_r$ are the dynamic matrices of the reset controller, $-1<\gamma<1$ determines the value of the reset state after the reset action, $r(t)\in\mathbb{R}$ is the reference signal, $y(t)\in\mathbb{R}$ is the output of the plant, and $e(t)=r(t)-y(t)$ is the tracking error. The focus of this paper is on GFORE and PCI, which have been mostly used in practice.
 In the case of GFORE one has
 \begin{equation}\label{GFO}
 C_R(s)=\dfrac{1}{\displaystyle\frac{s}{\omega_r}+1},
 \end{equation}	
 whereas for PCI one has
 \begin{equation}\label{GFO}
 C_R(s)=1+\dfrac{\omega_r}{s}.
 \end{equation}
   Thus, for GFORE, $A_r=-C_r=-\omega_r$ ($\omega_r$ is the so-called corner frequency), $D_r=0$ and $B_r=1$, whereas for the PCI, $A_r=0$, $C_r=\omega_r$ and $B_r=D_r=1$. 
   \newline Let now $\mathcal{L}(s)=C_L(s)G(s)$ and assume that $G(s)$ is strictly proper. Let the state-space realization of $\mathcal{L}(s)$ be
 \begin{equation}\label{E-22}
 \begin{cases} 
 \dot{\zeta}(t)=A\zeta(t)+Bu_r(t)+B_dd(t),\\
 y(t)=C\zeta(t),
 \end{cases}
 \end{equation}
 where $\zeta(t)\in\mathbb{R}^{n_p}$ describes the state of the plant and of the linear controller ($n_p$ is the number of states of the whole linear part), $A$, $B$, and $C$ are the dynamic matrices, and $d(t)\in\mathbb{R}$ is an external disturbance. The closed-loop state-space representation of the overall system can, therefore, be written as
 \begin{equation}\label{E-23}
 \begin{cases} 
 \dot{x}(t)=\bar{A}x(t)+\bar{B}r(t)+\bar{B}_dd(t), & e(t)\neq 0,\\
 x(t^+)=\bar{A}_\rho x(t), & e(t)=0,  \\
 y(t)=\bar{C}x(t),
 \end{cases}
 \end{equation} 
 where $x(t)=[x_r(t)^T\quad \zeta(t)^T]^T\in\mathbb{R}^{n_p+1}$, and \newline\newline 
 $\bar{A}=\begin{bmatrix} A_r & -B_rC \\ BC_r & A-BD_rC\end{bmatrix}$,             $\quad\quad\bar{B}=\begin{bmatrix} 1 \\ D_rB \end{bmatrix}$,  $\quad\quad\bar{B_d}=\begin{bmatrix} 0 \\ B_d \end{bmatrix},$\newline\newline$\bar{A}_\rho=\begin{bmatrix}\gamma & 0 \\ 0 & I_{n_p\times n_p} \end{bmatrix}$, and $\bar{C}=\begin{bmatrix} 0 & C \end{bmatrix}$.    
 The main goal of the paper is to provide  frequency-domain sufficient conditions to assess UBIBS stability of the reset control system (\ref{E-23}) with the control structure depicted in Fig. \ref{F-21}. 
 \section{Frequency-domain stability analysis}\label{sec:3}  
 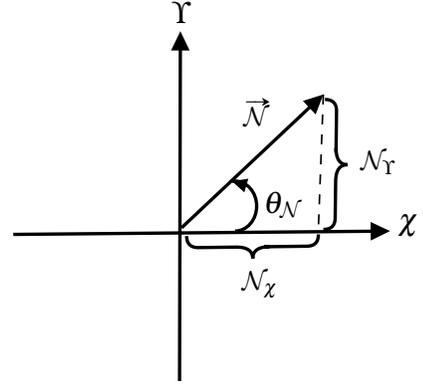
\begin{figure}[!t]
 	\centering
	\resizebox{0.65\columnwidth}{!}{%
 	\tikzset{every picture/.style={line width=0.75pt}} 
 	\begin{tikzpicture}[x=0.75pt,y=0.75pt,yscale=-1,xscale=1]
 	\draw [line width=1.5]    (9,150) -- (231,148.04) ;
 	\draw [shift={(235,148)}, rotate = 539.49] [fill={rgb, 255:red, 0; green, 0; blue, 0 }  ][line width=0.08]  [draw opacity=0] (11.61,-5.58) -- (0,0) -- (11.61,5.58) -- cycle    ;
 	\draw [line width=1.5]    (108.75,238.01) -- (109.24,32.01) ;
 	\draw [shift={(109.25,28.01)}, rotate = 450.14] [fill={rgb, 255:red, 0; green, 0; blue, 0 }  ][line width=0.08]  [draw opacity=0] (11.61,-5.58) -- (0,0) -- (11.61,5.58) -- cycle    ;
 	\draw [line width=1.5]    (110,146) -- (192.09,68.74) ;
 	\draw [shift={(195,66)}, rotate = 496.74] [fill={rgb, 255:red, 0; green, 0; blue, 0 }  ][line width=0.08]  [draw opacity=0] (11.61,-5.58) -- (0,0) -- (11.61,5.58) -- cycle    ;
 	\draw  [draw opacity=0][line width=1.5]  (140.24,118.13) .. controls (140.81,118.04) and (141.4,117.99) .. (141.99,117.99) .. controls (149.16,117.98) and (154.99,124.69) .. (155,132.97) .. controls (155.02,141.26) and (149.22,147.99) .. (142.04,148.01) .. controls (142.03,148.01) and (142.02,148.01) .. (142,148.01) -- (142.02,133) -- cycle ; \draw  [line width=1.5]  (140.24,118.13) .. controls (140.81,118.04) and (141.4,117.99) .. (141.99,117.99) .. controls (149.16,117.98) and (154.99,124.69) .. (155,132.97) .. controls (155.02,141.26) and (149.22,147.99) .. (142.04,148.01) .. controls (142.03,148.01) and (142.02,148.01) .. (142,148.01) ;
 	\draw    (150,121) -- (143.12,118.98) ;
 	\draw [shift={(140.24,118.13)}, rotate = 376.40999999999997] [fill={rgb, 255:red, 0; green, 0; blue, 0 }  ][line width=0.08]  [draw opacity=0] (8.93,-4.29) -- (0,0) -- (8.93,4.29) -- cycle    ;
 	\draw  [dash pattern={on 4.5pt off 4.5pt}]  (195,66) -- (192,147) ;
 	\draw  [line width=1.5]  (195,147) .. controls (199.67,147.06) and (202.03,144.76) .. (202.09,140.09) -- (202.37,118.09) .. controls (202.46,111.42) and (204.83,108.12) .. (209.5,108.18) .. controls (204.83,108.12) and (202.54,104.76) .. (202.62,98.09)(202.59,101.09) -- (202.9,76.09) .. controls (202.96,71.42) and (200.66,69.06) .. (195.99,69) ;
 	\draw  [line width=1.5]  (113,151) .. controls (113.06,155.67) and (115.42,157.97) .. (120.09,157.91) -- (142.59,157.62) .. controls (149.26,157.54) and (152.62,159.83) .. (152.68,164.5) .. controls (152.62,159.83) and (155.92,157.46) .. (162.59,157.37)(159.59,157.41) -- (185.09,157.08) .. controls (189.76,157.02) and (192.06,154.66) .. (192,149.99) ;
 
 	\draw (246,145) node  [font=\Large]  {$\chi $};
 	\draw (110,17) node  [font=\Large]  {$\Upsilon $};
 	\draw (171,132) node   [font=\Large] {$\theta _{\mathcal{N}}$};
 	\draw (155,77) node  [font=\large]  {$\vv{\mathcal{N}}$};
 	\draw (228,107) node  [font=\large]  {$\mathcal{N}_{\Upsilon }$};
 	\draw (157,180) node  [font=\large]  {$\mathcal{N}_{\chi }$};
 	\end{tikzpicture}}
 	\caption{Representation of the NSV in the $\chi-\Upsilon$ plane}
 	\label{F-31}
 \end{figure}
 In this section the main results, which are based on the so-called $H_{\beta}$-condition \cite{beker2004fundamental,guo2015analysis,banos2011reset}, are presented. Let 
 \begin{equation}\label{E-31}
 C_0=[\rho\quad\beta C],\quad B_0=\begin{bmatrix} 1 \\ 0_{n_p\times 1} \end{bmatrix},\quad \rho>0,\quad\beta\in\mathbb{R}.
 \end{equation}
 The $H_\beta$ condition, in the case of the PCI and of the GFORE, states that the reset control system (\ref{E-23}) with $-1\leq\gamma\leq1$, and $r=d=0$ is quadratically stable if and only if there exist $\rho>0$ and $\beta$ such that the transfer function 
 \begin{equation}\label{E-311}
 H(s)=C_0(sI-\bar{A})^{-1}B_0
 \end{equation} 
 is Strictly Positive Real (SPR). This condition requires finding the parameters $\rho$ and $\beta$, which may be very difficult when the system has a high order transfer function. In the following, a method to determine stability without finding $\rho$ and $\beta$ is proposed.  
 
 To this end, define the Nyquist Stability Vector (NSV=$\vv{\mathcal{N}}(\omega)\in\mathbb{R}^2$) in a plane with axis $\chi-\Upsilon$ (see Fig. \ref{F-31}) as follows. 
 \begin{definition}\label{D1}
 	The Nyquist Stability Vector is, for all $\omega\in\mathbb{R}^+$, the vector 
	$$\vv{\mathcal{N}}(\omega)=[\mathcal{N}_\chi \quad \mathcal{N}_\Upsilon]^T=$$  
	$$[\abs{L(j\omega)+\frac{1}{2}}^2-\frac{1}{4},\ \quad \Re(L(j\omega)\cdot C_R(j\omega))+\Re(C_R(j\omega))]^T,$$
$L(s)=\mathcal{L}(s)C_R(s)$.
  \end{definition}
Let, for simplicity and without loss of generality, $\phase{\vv{\mathcal{N}}(\omega)}=\theta_{\mathcal{N}}\in[-\frac{\pi}{2},\ \frac{3\pi}{2})$, and define the open sets
$$\mathcal{I}_1=\left\{\omega\in\mathbb{R}^+|\ 0<\phase{\vv{\mathcal{N}}(\omega)}<\frac{\pi}{2}\right\},$$ 
$$\mathcal{I}_2=\left\{\omega\in\mathbb{R}^+|\ \dfrac{\pi}{2}<\phase{\vv{\mathcal{N}}(\omega)}<\pi\right\},$$
$$\mathcal{I}_3=\left\{\omega\in\mathbb{R}^+|\ \pi<\phase{\vv{\mathcal{N}}(\omega)}<\dfrac{3\pi}{2}\right\},$$
$$\mathcal{I}_4=\left\{\omega\in\mathbb{R}^+|\ -\dfrac{\pi}{2}<\phase{\vv{\mathcal{N}}(\omega)}<0\right\}.$$	
	Defice now the $H_\beta$ circle in the complex plane with centre $(-\frac{1}{2},\ 0)$ and radius $\frac{1}{2}$ (see Fig. \ref{F-32}). Then, the following statements hold.
	\begin{itemize}
		\item  For all $\omega$ such that $L(j\omega)$ is outside the $H_\beta$ circle $\mathcal{N}_\chi>0$.
		\item  For all $\omega$ such that $L(j\omega)$ is on the $H_\beta$ circle $\mathcal{N}_\chi=0$.
		\item  For all $\omega$ such that $L(j\omega)$ is inside the $H_\beta$ circle $\mathcal{N}_\chi<0$.
	\end{itemize} 		
On the basis of the definition of the NSV, systems of Type I and of Type II, which are used to assess the stability of the reset control systems, are defined.  

\begin{definition}\label{D2}
The reset control system (\ref{E-23}) is of Type I if the following conditions hold. 
\begin{enumerate}[(1)]
\item For all $\omega\in\mathcal{M}=\{\omega\in\mathbb{R}^+|\ \mathcal{N}_\chi(\omega)=0\}$ one has $\mathcal{N}_\Upsilon(\omega)>0$.
\item For all $\omega\in\mathcal{Q}=\{\omega\in\mathbb{R}^+|\ \mathcal{N}_\Upsilon(\omega)=0\}$ one has $\mathcal{N}_\chi(\omega)>0$. 
\item At least one of the following statements is true:
\begin{enumerate}
	\item $\forall\ \omega\in\mathbb{R}^+:\ \mathcal{N}_\Upsilon(\omega)\geq0,$
	\item $\forall\ \omega\in\mathbb{R}^+:\ \mathcal{N}_\chi(\omega)\geq0,$
	\item Let $\delta_1=\underset{\omega\in\mathcal{I}_4}{\max}\abs{\dfrac{\mathcal{N}_\Upsilon(\omega)}{\mathcal{N}_\chi(\omega)}}$ and $\Psi_1=\underset{\omega\in\mathcal{I}_2}{\min}\abs{\dfrac{\mathcal{N}_\Upsilon(\omega)}{\mathcal{N}_\chi(\omega)}}$. Then $\delta_1<\Psi_1$ and $\mathcal{I}_3=\varnothing$.
\end{enumerate}
	\end{enumerate}	
 \end{definition}

\begin{remark}\label{Rs1}
Let
\begin{equation}\label{E-333}
\theta_{1}=\underset{\omega\in\mathbb{R}^+}{\min}\phase{\vv{\mathcal{N}}(\omega)}=\phase{\vv{\mathcal{N}}_1}\text{ and }\theta_{2}=\underset{\omega\in\mathbb{R}^+}{\max}\phase{\vv{\mathcal{N}}(\omega)}=\phase{\vv{\mathcal{N}}_2},
\end{equation}
, where $\vv{\mathcal{N}}_1$ and $\vv{\mathcal{N}}_2$ are implicitly defined by equation~(\ref{E-333}). Then, the conditions identifying Type I systems are equivalent to the condition  
\begin{equation}\label{E-3333}
\left(-\dfrac{\pi}{2}<\theta_{1}<\pi\right)\ \land\ \left(-\dfrac{\pi}{2}<\theta_{2}<\pi\right)\ \land\ (\theta_{2}-\theta_{1}<\pi).
\end{equation}	
\end{remark}

\begin{definition}\label{D3}
The reset control system (\ref{E-23}) is of Type II if the following conditions hold: 
	\begin{enumerate}[(1)]
		\item $\mathcal{L}(s)$ does not have any pole at origin.
		\item For all $\omega\in\mathcal{M}$ one has $\mathcal{N}_\Upsilon(\omega)>0$.
		\item For all $\omega\in\mathcal{Q}$ one has $\mathcal{N}_\chi(\omega)<0$ 
		\item At least, one of the following statements is true:
		\begin{enumerate}
			\item $\forall\ \omega\in\mathbb{R}^+:\ \mathcal{N}_\Upsilon(\omega)\geq0$
			\item $\forall\ \omega\in\mathbb{R}^+:\ \mathcal{N}_\chi(\omega)\leq0$
			\item Let $\delta_2=\underset{\omega\in\mathcal{I}_3}{\max}\abs{\dfrac{\mathcal{N}_\Upsilon(\omega)}{\mathcal{N}_\chi(\omega)}}$ and $\Psi_2=\underset{\omega\in\mathcal{I}_1}{\min}\abs{\dfrac{\mathcal{N}_\Upsilon(\omega)}{\mathcal{N}_\chi(\omega)}}$. Then, $\delta_2<\Psi_2$ and $\mathcal{I}_4=\varnothing$.  
		\end{enumerate}
	\end{enumerate}	
\end{definition}		

\begin{remark}\label{Rs2}
	The conditions identifying the Type II systems are equivalent to the following conditions. 
		\begin{enumerate}[(1)]
		\item $\mathcal{L}(s)$ does not have any pole at origin.
	    \item \begin{equation}\label{E-555}
	\left(0<\theta_{1}<\dfrac{3\pi}{2}\right)\ \land\ \left(0<\theta_{2}<\dfrac{3\pi}{2}\right)\ \land\ (\theta_{2}-\theta_{1}<\pi).
	\end{equation}
	\end{enumerate}	
\end{remark}
 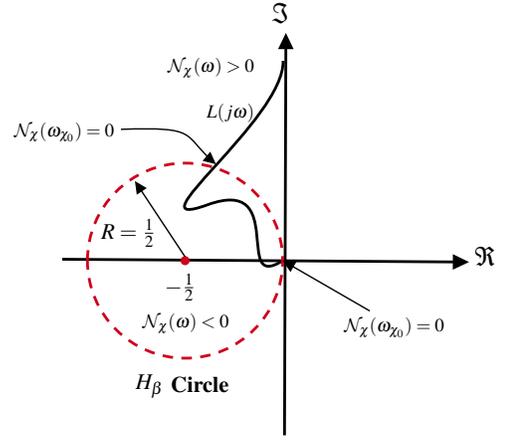
\begin{figure}[!t]
 	\centering
 	\resizebox{0.78\columnwidth}{!}{%
 		\tikzset{every picture/.style={line width=0.75pt}} 
 		\begin{tikzpicture}[x=0.75pt,y=0.75pt,yscale=-1,xscale=1]
 		\draw [line width=1.5]    (64.5,187) -- (351.5,188.97) ;
 		\draw [shift={(355.5,189)}, rotate = 180.39] [fill={rgb, 255:red, 0; green, 0; blue, 0 }  ][line width=0.08]  [draw opacity=0] (11.61,-5.58) -- (0,0) -- (11.61,5.58) -- cycle    ;
 		\draw [line width=1.5]    (223.5,313) -- (224.49,29) ;
 		\draw [shift={(224.5,25)}, rotate = 450.2] [fill={rgb, 255:red, 0; green, 0; blue, 0 }  ][line width=0.08]  [draw opacity=0] (11.61,-5.58) -- (0,0) -- (11.61,5.58) -- cycle    ;
 		\draw [line width=1.5]    (222.5,45) .. controls (223.5,87) and (110.5,169) .. (168.5,148) .. controls (226.5,127) and (186.19,211) .. (222.19,188) ;
 		\draw  [color={rgb, 255:red, 208; green, 2; blue, 27 }  ,draw opacity=1 ][dash pattern={on 5.63pt off 4.5pt}][line width=1.5]  (82.56,188) .. controls (82.56,149.44) and (113.82,118.19) .. (152.38,118.19) .. controls (190.93,118.19) and (222.19,149.44) .. (222.19,188) .. controls (222.19,226.56) and (190.93,257.81) .. (152.38,257.81) .. controls (113.82,257.81) and (82.56,226.56) .. (82.56,188) -- cycle ;
 		\draw  [color={rgb, 255:red, 208; green, 2; blue, 27 }  ,draw opacity=1 ][fill={rgb, 255:red, 208; green, 2; blue, 27 }  ,fill opacity=1 ] (149.63,188) .. controls (149.63,186.48) and (150.86,185.25) .. (152.38,185.25) .. controls (153.89,185.25) and (155.13,186.48) .. (155.13,188) .. controls (155.13,189.52) and (153.89,190.75) .. (152.38,190.75) .. controls (150.86,190.75) and (149.63,189.52) .. (149.63,188) -- cycle ;
 		\draw    (152.38,185.25) -- (118.15,133.5) ;
 		\draw [shift={(116.5,131)}, rotate = 416.52] [fill={rgb, 255:red, 0; green, 0; blue, 0 }  ][line width=0.08]  [draw opacity=0] (8.93,-4.29) -- (0,0) -- (8.93,4.29) -- cycle    ;
 		\draw [line width=0.75]    (106.5,94) .. controls (162.76,92.06) and (150.33,97.65) .. (172.35,117.15) ;
 		\draw [shift={(174.5,119)}, rotate = 220.03] [fill={rgb, 255:red, 0; green, 0; blue, 0 }  ][line width=0.08]  [draw opacity=0] (8.93,-4.29) -- (0,0) -- (8.93,4.29) -- cycle    ;
 		\draw    (284.5,224) -- (224.79,189.5) ;
 		\draw [shift={(222.19,188)}, rotate = 390.02] [fill={rgb, 255:red, 0; green, 0; blue, 0 }  ][line width=0.08]  [draw opacity=0] (8.93,-4.29) -- (0,0) -- (8.93,4.29) -- cycle    ;
 		
 		\draw (367,185) node  [font=\Large]  {$\Re $};
 		\draw (221,11) node  [font=\Large]  {$\Im $};
 		\draw (149,207) node  [scale=1.2]  {$-\frac{1}{2}$};
 		\draw (112,168) node  [scale=1.2]  {$R=\frac{1}{2}$};
 		\draw (171,50) node    {$\mathcal{N}_{\chi }( \omega )  >0$};
 		\draw (153,233) node    {$\mathcal{N}_{\chi }(\omega) < 0$};
 		\draw (66,96) node    {$\mathcal{N}_{\chi }( \omega_{\chi_0}) =0$};
 		\draw (302,236) node    {$\mathcal{N}_{\chi }(\omega_{\chi_0}) =0$};
 		\draw (185,82) node    {$L( j\omega )$};
 		\draw (127,278) node  [font=\large]  {$H_{\beta}$};
 		\draw (163,277) node   [align=left] {{\fontfamily{ptm}\selectfont {\large \textbf{Circle}}}};
 		\end{tikzpicture}
 	}
 	\caption{$H_\beta$ circle in the Nyquist diagram}
 	\label{F-32}
 \end{figure}
On the basis of the above definitions the main result of this paper, which is a frequency-domain tool for determining stability of reset control systems, is presented.
 \begin{theorem}\label{T1}
 	The reset control system (\ref{E-23}) with GFORE or PCI is UBIBS stable if all the following conditions are satisfied.	
 	\begin{itemize}
 		\item The base linear system is stable and the open-loop transfer function does not have any pole-zero cancellation.
 		\item The reset control system (\ref{E-23}) is either of Type I and/or of Type II. 
 	\end{itemize}
 \end{theorem}
\begin{proof}
 	 Theorem \ref{T1} is proved in several steps.
 	\begin{itemize}
 		\item Step 1: It is shown that, by Hypothesis (II) of Theorem \ref{T1}, it is possible to find $\beta$ and $\rho>0$ such that $\Re(H(j\omega))>0,\ \forall\ \omega\in\mathbb{R}^+$. 
 		\item Step 2: For systems with poles at origin, it is shown that $\displaystyle\lim_{\omega\to 0} \Re(H(j\omega))>0$.
 		\item Step 3: It is shown that either $\displaystyle\lim_{s\to \infty} H(s)>0$ or $\displaystyle\lim_{\omega\to \infty} \omega^2\Re(H(j\omega))>0.$
 		\item Step 4: It is shown that $(A,C_0)$ and $(A,B_0)$ are observable and controllable, respectively.
 		\item Step 5: It is concluded that $H(s)$ is SPR and the $H_\beta$ condition is satisfied. Then it is proved that for GFORE or PCI, the reset control system (\ref{E-23}) is UBIBS stable.
 	\end{itemize}  
Step 1: The transfer function (\ref{E-311}) can be rewritten as 
\begin{equation}\label{E-32}
H(s)=\dfrac{y}{r}=\dfrac{\beta L(s)+\rho^\prime C_R(s)}{1+L(s)}, \text{ (see also Fig. \ref{F-33})}.
\end{equation}
Let $L(j\omega)=a+bj$ and $C_R(j\omega)=a_R+b_Rj$. Then, 	
\begin{equation}\label{E-33}
\Re(H(j\omega))=\dfrac{\beta\left((a+\frac{1}{2})^2+b^2-\frac{1}{4}\right)+\rho^\prime\left(a_Ra+b_rb+a_R\right)}{(a+1)^2+b^2}.
\end{equation}
Define now the vector $\vv{\xi}\in\mathbb{R}^2$ as $\vv{\xi}=[\beta\quad\rho]^T$ in the $\chi-\Upsilon$ plane. Using Definition \ref{D1}, equation (\ref{E-33}) can be re-written as 
\begin{equation}\label{E-34}
\Re(H(j\omega))=\dfrac{\vv{\xi}\cdot\vv{\mathcal{N}}}{(a+1)^2+b^2}.
\end{equation}
Then, the $H_\beta$ condition reduces to 
\begin{equation}\label{E-36}
\begin{array}{*{35}{c}}
\forall\omega\in\mathbb{R}^+:\ \Re(H(j\omega))>0\iff\vv{\xi}\cdot\vv{\mathcal{N}}>0\iff\\
-\frac{\pi}{2}<\phase{(\vv{\xi},\vv{\mathcal{N}})}<\frac{\pi}{2}\ \land\ \abs{\vv{\mathcal{N}}}\neq 0\ \land\ \abs{\vv{\mathcal{\xi}}}\neq 0.
\end{array}
\end{equation}
By (\ref{E-333}), $\forall\omega\in\mathbb{R}^+$, $\vv{\mathcal{N}}(\omega)$ is placed between the vectors $\vv{\mathcal{N}_1}$ and $\vv{\mathcal{N}_2}$ illustrated in Fig. \ref{F-34}. In other words,
\begin{equation}\label{E-37}
\forall\ \omega\in\mathbb{R}^+:\ \theta_1\leq\phase{\vv{\mathcal{N}}(\omega)}\leq\theta_2.
\end{equation}   
 If $\beta>0$, since $0<\phase{\vv{\xi}}=\theta_\xi<\frac{\pi}{2}$, then $\theta_1\in(-\dfrac{\pi}{2},\pi)$ and $\theta_2\in(-\dfrac{\pi}{2},\pi)$. This implies the conditions (1) and (2) in Definition~\ref{D2} and $\mathcal{I}_3=\varnothing$. If $\beta\leq0$, then $\theta_1\in(0,\dfrac{3\pi}{2})$ and $\theta_2\in(0,\dfrac{3\pi}{2})$, which implies the conditions (1) and (2) in Definition~\ref{D3} hold and $\mathcal{I}_4=\varnothing$. If $\theta_1\in[0,\dfrac{\pi}{2}]$ and $\theta_2\in[0,\dfrac{\pi}{2}]$, then
\begin{equation}\label{E-39}
\Re(H(j\omega))>0\iff
\begin{cases}
\theta_{\xi }\in(0,\dfrac{\pi}{2})\iff\beta>0,   \\
\theta_{\xi }\in[\dfrac{\pi}{2},\dfrac{\pi}{2}+\theta_{1})\Rightarrow\beta\leq 0\ \land\ \theta_1>0.\\
\end{cases}
\end{equation}
If $\theta_1\in[0,\dfrac{\pi}{2}]$ and $\theta_2\in[\dfrac{\pi}{2},\pi]$, then
\begin{equation}\label{E-40}
\Re(H(j\omega))>0\iff
\begin{cases}
\theta_{\xi }\in(\theta_{2}-\dfrac{\pi}{2},\dfrac{\pi}{2})\Rightarrow\beta>0\land\ \theta_2<\pi,   \\
\theta_{\xi }\in[\dfrac{\pi}{2},\dfrac{\pi}{2}+\theta_{1})\Rightarrow\beta\leq 0\ \land\ \theta_1>0.\\
	\end{cases}
\end{equation}
If $\theta_1\in[\dfrac{\pi}{2},\pi]$ and $\theta_2\in[\dfrac{\pi}{2},\pi]$, then 
\begin{equation}\label{E-41}
\Re(H(j\omega))>0\iff
	\begin{cases}
\theta_{\xi }\in(\theta_{2}-\dfrac{\pi}{2},\dfrac{\pi}{2})\Rightarrow\beta>0\land\ \theta_2<\pi,   \\
\theta_{\xi }\in[\dfrac{\pi}{2},\pi)\iff\beta\leq 0.\\
	\end{cases}
\end{equation}
If $\theta_1\in[\dfrac{\pi}{2},\dfrac{3\pi}{2})$ and $\theta_2\in[\pi,\dfrac{3\pi}{2})$, then
\begin{equation}\label{E-42}
\Re(H(j\omega))>0\iff\theta_{\xi}\in(\theta_{2}-\dfrac{\pi}{2},\pi)\Rightarrow\beta<0.
\end{equation}
If $\theta_1\in(0,\dfrac{\pi}{2}]$ and $\theta_2\in[\pi,\dfrac{3\pi}{2})$, then $\Re(H(j\omega))>0$ if and only if
\begin{equation}\label{E-45}
\left(\theta_{\xi}\in(\theta_{2}-\dfrac{\pi}{2},\theta_{1}+\dfrac{\pi}{2})\ \land\ \theta_{2}-\theta_{1}<\pi\right)\Rightarrow\beta<0.
\end{equation}
As a result.
\begin{equation}\label{E-46}
\theta_{2}-\theta_{1}<\pi\iff\delta_2<\psi_2.
\end{equation}
Hence, by (\ref{E-39})-(\ref{E-46}), Condition (3) of Definition \ref{D3} and Condition (2) of Remark \ref{Rs2} are obtained. If $\theta_1\in(-\dfrac{\pi}{2},0]$ and $\theta_2\in(-\dfrac{\pi}{2},\dfrac{\pi}{2}]$, then
\begin{equation}\label{E-48}
\Re(H(j\omega))>0\iff\theta_{\xi}\in(0,\theta_{1}+\dfrac{\pi}{2})\Rightarrow\beta>0.
\end{equation}
If $\theta_1\in(-\dfrac{\pi}{2},0]$ and $\theta_2\in[\dfrac{\pi}{2},\pi)$, then $\Re(H(j\omega))>0$ if and only if
\begin{equation}\label{E-49}
\left(\theta_{\xi}\in(\theta_{2}-\dfrac{\pi}{2},\theta_{1}+\dfrac{\pi}{2})\ \land\ \theta_{2}-\theta_{1}<\pi\right)\Rightarrow\beta>0,
\end{equation}
hence 
\begin{equation}\label{E-50}
\theta_{2}-\theta_{1}<\pi\iff\delta_1<\psi_1.
\end{equation}
Therefore, by (\ref{E-39})-(\ref{E-41}) and (\ref{E-48})-(\ref{E-50}), Condition (3) of Definition \ref{D2} and Remark \ref{Rs1} are obtained.
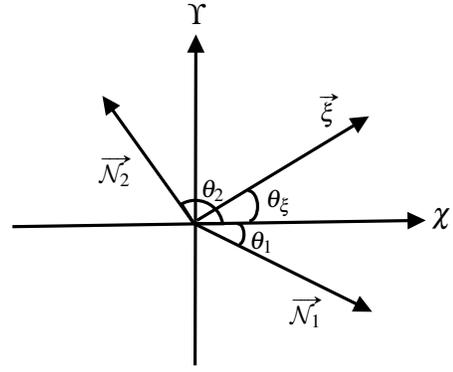
\begin{figure}[!t]
	\centering
	\resizebox{0.7\columnwidth}{!}{%
\tikzset{every picture/.style={line width=0.75pt}} 
\begin{tikzpicture}[x=0.75pt,y=0.75pt,yscale=-1,xscale=1]
\draw [line width=1.5]    (7,146) -- (265,142.06) ;
\draw [shift={(269,142)}, rotate = 539.13] [fill={rgb, 255:red, 0; green, 0; blue, 0 }  ][line width=0.08]  [draw opacity=0] (11.61,-5.58) -- (0,0) -- (11.61,5.58) -- cycle    ;
\draw [line width=1.5]    (122.75,234.01) -- (123.24,28.01) ;
\draw [shift={(123.25,24.01)}, rotate = 450.14] [fill={rgb, 255:red, 0; green, 0; blue, 0 }  ][line width=0.08]  [draw opacity=0] (11.61,-5.58) -- (0,0) -- (11.61,5.58) -- cycle    ;
\draw [line width=1.5]    (124,142) -- (229.58,78.07) ;
\draw [shift={(233,76)}, rotate = 508.8] [fill={rgb, 255:red, 0; green, 0; blue, 0 }  ][line width=0.08]  [draw opacity=0] (11.61,-5.58) -- (0,0) -- (11.61,5.58) -- cycle    ;
\draw [line width=1.5]    (122,144) -- (231.42,198.22) ;
\draw [shift={(235,200)}, rotate = 206.36] [fill={rgb, 255:red, 0; green, 0; blue, 0 }  ][line width=0.08]  [draw opacity=0] (11.61,-5.58) -- (0,0) -- (11.61,5.58) -- cycle    ;
\draw [line width=1.5]    (122,144) -- (66.33,66.25) ;
\draw [shift={(64,63)}, rotate = 414.4] [fill={rgb, 255:red, 0; green, 0; blue, 0 }  ][line width=0.08]  [draw opacity=0] (11.61,-5.58) -- (0,0) -- (11.61,5.58) -- cycle    ;
\draw  [draw opacity=0][line width=1.5]  (156.09,123.25) .. controls (156.19,123.24) and (156.29,123.23) .. (156.39,123.23) .. controls (159.49,123.1) and (162.17,127.28) .. (162.38,132.56) .. controls (162.6,137.85) and (160.26,142.23) .. (157.16,142.36) .. controls (157.15,142.36) and (157.14,142.36) .. (157.13,142.36) -- (156.78,132.79) -- cycle ; \draw  [line width=1.5]  (156.09,123.25) .. controls (156.19,123.24) and (156.29,123.23) .. (156.39,123.23) .. controls (159.49,123.1) and (162.17,127.28) .. (162.38,132.56) .. controls (162.6,137.85) and (160.26,142.23) .. (157.16,142.36) .. controls (157.15,142.36) and (157.14,142.36) .. (157.13,142.36) ;
\draw  [draw opacity=0][line width=1.5]  (149.32,144.38) .. controls (149.55,144.28) and (149.8,144.21) .. (150.05,144.2) .. controls (151.94,144.13) and (153.6,146.96) .. (153.74,150.53) .. controls (153.88,154.1) and (152.46,157.06) .. (150.57,157.14) .. controls (150.56,157.14) and (150.56,157.14) .. (150.55,157.14) -- (150.31,150.67) -- cycle ; \draw  [line width=1.5]  (149.32,144.38) .. controls (149.55,144.28) and (149.8,144.21) .. (150.05,144.2) .. controls (151.94,144.13) and (153.6,146.96) .. (153.74,150.53) .. controls (153.88,154.1) and (152.46,157.06) .. (150.57,157.14) .. controls (150.56,157.14) and (150.56,157.14) .. (150.55,157.14) ;
\draw  [line width=1.5]  (122,144) -- (114.15,131.67) .. controls (116.73,130.18) and (119.72,129.32) .. (122.91,129.32) .. controls (131.65,129.32) and (138.89,135.73) .. (140.2,144.09) -- (122,144) -- cycle ;
\draw (279,141) node  [font=\Large]  {$\chi $};
\draw (124,11) node  [font=\Large]  {$\Upsilon $};
\draw (207,73) node  [font=\large]  {$\vv{\xi }$};
\draw (193,201) node  [font=\large]  {$\vv{\mathcal{N}_{1}}$};
\draw (72,111) node  [font=\large]  {$\vv{\mathcal{N}_{2}}$};
\draw (165,156) node [scale=1.2]   {$\theta _{1}$};
\draw (176,131) node [scale=1.2]   {$\theta _{\xi }$};
\draw (134,123) node [scale=1.2]  {$\theta _{2}$};
\end{tikzpicture}}
	\caption{Representation of $\vec{\mathcal{N}}(\omega)$ and $\vec{\xi}$ in the $\chi-\Upsilon$ plane}
\label{F-34}
\end{figure}
\newline Step 2: Let $L(s)=\dfrac{L^\prime(s)}{s^n},\ \text{with }n \geq 1,\ L^\prime(0)\neq 0$. Equation~(\ref{E-33}) yields
\begin{equation}\label{E-43}
\resizebox{\columnwidth}{!}{$
	\displaystyle\lim_{\omega\to 0}\Re(H(j\omega))=\displaystyle\lim_{|L|\to \infty}\dfrac{\beta|L|^2+\rho^\prime\left(|L||C_R(0)|\cos(\phase{\vv{C_R}(0),\vv{L(0)}})+\Re(C_R(0))\right)}{|L|^2}.$}
\end{equation}
For GFORE, equation (\ref{E-43}) becomes
\begin{equation}\label{E-44}
\resizebox{\columnwidth}{!}{$
	\displaystyle\lim_{\omega\to 0}\Re(H(j\omega))=\beta+\rho^\prime\displaystyle\lim_{|L|\to \infty}\dfrac{\cos(\phase{\vv{C_R}(0),\vv{L(0)}})}{|L|}+\dfrac{1}{|L|^2}=\beta>0,$}
\end{equation}
whereas in the case of PCI with $n=1$ (\ref{E-43}) becomes 
\begin{equation}\label{E-445}
\resizebox{\columnwidth}{!}{$
	\displaystyle\lim_{\omega\to 0}\Re(H(j\omega))=\beta+\rho^\prime\displaystyle\lim_{\omega\to 0}\left(\dfrac{|C_R|}{|L|}+\dfrac{1}{|L|^2}\right)=\beta+\dfrac{\rho^\prime\omega_r}{|\mathcal{L}(0)|}$}
\end{equation}
which, setting $\vv{\mathcal{N}^{\prime}}=[1\quad\dfrac{\rho^\prime\omega_r}{|\mathcal{L}(0)|}]^T$, yields
\begin{equation}\label{E-447}
\displaystyle\lim_{\omega\to 0}\Re(H(j\omega))=\vv{\xi}\cdot\vv{\mathcal{N}^{\prime}}.
\end{equation}  
In addition,
\begin{equation}\label{E-471}
\phase{\vv{\mathcal{N}^{\prime}}}=\displaystyle\lim_{\omega\to 0}\phase{\vv{\mathcal{N}}}\xRightarrow[]{\ (\ref{E-37})\ }\theta_{1}\leq\phase{\vv{\mathcal{N}^{\prime}}}\leq\theta_{2}.
\end{equation}
As a result, by Step 1, $\displaystyle\lim_{\omega\to 0} \Re(H(j\omega))=\vv{\xi}\cdot\vv{\mathcal{N}^{\prime}}>0$. For PCI with $n>1$
\begin{equation}\label{E-446}
\displaystyle\lim_{\omega\to 0}\Re(H(j\omega))=\beta+\rho^\prime\displaystyle\lim_{\omega\to 0}\dfrac{\omega^n\cos(\phase{\vv{C_R}(0),\vv{L(0)}})}{\omega}=\beta>0.
\end{equation}
It is therefore concluded that for systems with poles at the origin (i.e. $\mathcal{L}(s)=\dfrac{\mathcal{L}^\prime(s)}{s^n},\ n \geq 1,\ \mathcal{L}^\prime(0)\neq 0$), $\beta>0$. If $\mathcal{L}(s)$ does not have any pole at origin, $\beta$ can be either positive or negative.    
\begin{figure}[!t]
	\centering
	\resizebox{0.9\columnwidth}{!}{%
		\tikzset{every picture/.style={line width=0.75pt}} 
		\begin{tikzpicture}[x=0.75pt,y=0.75pt,yscale=-1,xscale=1]
		\draw  [fill={rgb, 255:red, 241; green, 241; blue, 241 }  ,fill opacity=1 ][dash pattern={on 1.69pt off 2.76pt}][line width=1.5]  (241,98.45) .. controls (241,87.4) and (249.95,78.45) .. (261,78.45) -- (409,78.45) .. controls (420.05,78.45) and (429,87.4) .. (429,98.45) -- (429,158.45) .. controls (429,169.5) and (420.05,178.45) .. (409,178.45) -- (261,178.45) .. controls (249.95,178.45) and (241,169.5) .. (241,158.45) -- cycle ;
		\draw  [line width=1.5]  (264.5,82) -- (343,82) -- (343,148) -- (264.5,148) -- cycle ;
		\draw  [line width=1.5]  (48.63,112.45) .. controls (48.63,104.84) and (55.3,98.68) .. (63.53,98.68) .. controls (71.76,98.68) and (78.43,104.84) .. (78.43,112.45) .. controls (78.43,120.06) and (71.76,126.22) .. (63.53,126.22) .. controls (55.3,126.22) and (48.63,120.06) .. (48.63,112.45) -- cycle ;
		\draw [line width=1.5]    (439,116) -- (440,217) -- (69.01,219.01) -- (65,219) -- (65.51,130.22) ;
		\draw [shift={(65.53,126.22)}, rotate = 450.33] [fill={rgb, 255:red, 0; green, 0; blue, 0 }  ][line width=0.08]  [draw opacity=0] (11.61,-5.58) -- (0,0) -- (11.61,5.58) -- cycle    ;
		\draw [line width=1.5]    (5.5,114) -- (44.63,114.41) ;
		\draw [shift={(48.63,114.45)}, rotate = 180.6] [fill={rgb, 255:red, 0; green, 0; blue, 0 }  ][line width=0.08]  [draw opacity=0] (11.61,-5.58) -- (0,0) -- (11.61,5.58) -- cycle    ;
		\draw [line width=1.5]    (421,116) -- (461,116) ;
		\draw [shift={(465,116)}, rotate = 180] [fill={rgb, 255:red, 0; green, 0; blue, 0 }  ][line width=0.08]  [draw opacity=0] (11.61,-5.58) -- (0,0) -- (11.61,5.58) -- cycle    ;
		\draw [line width=1.5]    (344.49,116.5) -- (369.5,116.93) ;
		\draw [shift={(373.5,117)}, rotate = 180.99] [fill={rgb, 255:red, 0; green, 0; blue, 0 }  ][line width=0.08]  [draw opacity=0] (11.61,-5.58) -- (0,0) -- (11.61,5.58) -- cycle    ;
		\draw  [line width=1.5]  (372.17,88.09) -- (420.5,88.09) -- (420.5,146) -- (372.17,146) -- cycle ;
		\draw [line width=1.5]    (79.43,115.45) -- (118.5,115.04) ;
		\draw [shift={(122.5,115)}, rotate = 539.4] [fill={rgb, 255:red, 0; green, 0; blue, 0 }  ][line width=0.08]  [draw opacity=0] (11.61,-5.58) -- (0,0) -- (11.61,5.58) -- cycle    ;
		\draw [line width=1.5]    (204.49,115.5) -- (258,115.97) ;
		\draw [shift={(262,116)}, rotate = 180.5] [fill={rgb, 255:red, 0; green, 0; blue, 0 }  ][line width=0.08]  [draw opacity=0] (11.61,-5.58) -- (0,0) -- (11.61,5.58) -- cycle    ;
		\draw  [line width=1.5]  (163,85.09) -- (163,89.09) -- (203,89.09) -- (203,146) -- (162.17,146) -- (162.17,142) -- (122.17,142) -- (122.17,85.09) -- (163,85.09) -- cycle ;
		\draw  [dash pattern={on 4.5pt off 4.5pt}]  (100.97,115.23) -- (101,133) -- (122,133) ;
		\draw [line width=1.5]    (221,114) -- (222,32) -- (335,32) ;
		\draw [shift={(339,32)}, rotate = 180] [fill={rgb, 255:red, 0; green, 0; blue, 0 }  ][line width=0.08]  [draw opacity=0] (11.61,-5.58) -- (0,0) -- (11.61,5.58) -- cycle    ;
		\draw  [line width=1.5]  (340,5.09) -- (424,5.09) -- (424,59) -- (340,59) -- cycle ;
		\draw  [line width=1.5]  (518.63,116.45) .. controls (518.63,108.84) and (525.3,102.68) .. (533.53,102.68) .. controls (541.76,102.68) and (548.43,108.84) .. (548.43,116.45) .. controls (548.43,124.06) and (541.76,130.22) .. (533.53,130.22) .. controls (525.3,130.22) and (518.63,124.06) .. (518.63,116.45) -- cycle ;
		\draw [line width=1.5]    (425,31) -- (533,32) -- (533.5,98.68) ;
		\draw [shift={(533.53,102.68)}, rotate = 269.57] [fill={rgb, 255:red, 0; green, 0; blue, 0 }  ][line width=0.08]  [draw opacity=0] (11.61,-5.58) -- (0,0) -- (11.61,5.58) -- cycle    ;
		\draw  [line width=1.5]  (464.17,98) -- (492,98) -- (492,133) -- (464.17,133) -- cycle ;
		\draw [line width=1.5]    (493.65,117.05) -- (515,117.01) ;
		\draw [shift={(519,117)}, rotate = 539.88] [fill={rgb, 255:red, 0; green, 0; blue, 0 }  ][line width=0.08]  [draw opacity=0] (11.61,-5.58) -- (0,0) -- (11.61,5.58) -- cycle    ;
		\draw [line width=1.5]    (548.43,116.45) -- (588,116.95) ;
		\draw [shift={(592,117)}, rotate = 180.72] [fill={rgb, 255:red, 0; green, 0; blue, 0 }  ][line width=0.08]  [draw opacity=0] (11.61,-5.58) -- (0,0) -- (11.61,5.58) -- cycle    ;
		
		\draw (63.53,113.61) node  [scale=1.5,font=\large]  {$-$};
		\draw (25.84,100) node  [scale=1.5,font=\large]  {$r$};
		\draw (567.62,97) node  [scale=1.5,font=\large]  {$y$};
		\draw (165,114) node  [scale=1.5,font=\large]  {$\mathbf{C_{\mathbf{R}}}$};
		\draw (303.75,115) node  [scale=1.5,font=\large]  {$\mathbf{C_{\mathbf{L}}}$};
		\draw (396.34,116.04) node  [scale=1.5,font=\large]  {$\mathbf{G}$};
		\draw (255,164) node  [scale=1.5,font=\large]  {$\mathcal{L}$};
		\draw (382,32.04) node  [scale=1.4,font=\large]  {$\mathbf{\rho ^{\prime } =\dfrac{\rho }{C_{r}}}$};
		\draw (532.53,117.61) node  [scale=1.5,font=\large]  {$+$};
		\draw (478.09,114) node  [scale=1.5,font=\large]  {$\mathbf{\beta }$};
		\end{tikzpicture}}
	\caption{The block diagram representative of $H(s)$}
	\label{F-33}
\end{figure}
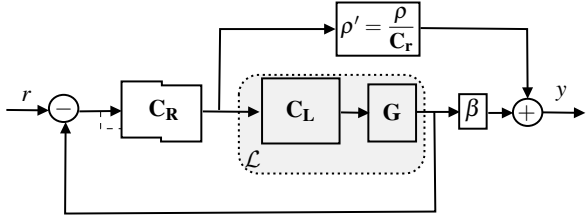
As a result, by Step 1 and Step 2, if Hypothesis (II) holds
\begin{equation}\label{E-51}
\exists \ (\beta\in\mathbb{R},\ \rho^\prime>0)\ |\ \forall\ \omega\in\mathbb{R}^+:\ \Re(H(j\omega))>0,
\end{equation}
and also, the claims in Remark \ref{Rs1} and Remark \ref{Rs2} are true. 
\newline Step 3: Since $\mathcal{L}(s)$ is strictly proper, it is possible to consider $\displaystyle\lim_{\omega\to \infty}|L|=\dfrac{|a_\infty+jb_\infty|}{|\omega|^n},\ n\geq2$. For GFORE, $|C_R|\approx\dfrac{\omega_r}{|\omega|}\text{ and } a_R\approx\dfrac{\omega_r^2}{\omega^2}$ for $\omega$ sufficiently large, hence, for $n=2$ and setting $\vv{\mathcal{N}^{\prime\prime}}=[a_\infty\quad\omega_r^2]^T$, yields
\begin{equation}\label{E-52}
\displaystyle\lim_{\omega\to \infty} \omega^2\Re(H(j\omega))=(\beta a_\infty+\rho^\prime\omega_r^2)=\vv{\xi}\cdot\vv{\mathcal{N}^{\prime\prime}}.
\end{equation}  
In addition
\begin{equation}\label{E-471}
\phase{\vv{\mathcal{N}^{\prime\prime}}}=\displaystyle\lim_{\omega\to\infty}\phase{\vv{\mathcal{N}}}\xRightarrow[]{\ (\ref{E-37})\ }\theta_{1}\leq\phase{\vv{\mathcal{N}^{\prime\prime}}}\leq\theta_{2}.
\end{equation}
Thus, by Step 1, $\displaystyle\lim_{\omega\to \infty} \omega^2\Re(H(j\omega))=\vv{\xi}\cdot\vv{\mathcal{N}^{\prime\prime}}>0$. For GFORE with $n>2$, $\displaystyle\lim_{\omega\to \infty}\omega^2\Re(H(j\omega))=\rho^\prime\omega_r^2>0$.  
For PCI, $\displaystyle\lim_{s\to \infty}H(s)=\rho>0$. Hence, by Hypothesis (II), $\displaystyle\lim_{s\to \infty}H(s)>0$ or $\displaystyle\lim_{\omega\to \infty}\omega^2\Re(H(j\omega))>0$.  
\newline Step 4: In order to show that the pairs $(A,C_0)$ and $(A,B_0)$ are observable and controllable, respectively, it is sufficient to show that the denominator and the numerator of $H(s)$ do not have any common root. Let $a_0+jb_0$ be a root of the denominator. Then
\begin{equation}\label{E-50-51-1}
\resizebox{\columnwidth}{!}{$
1+R_L(a_0,b_0)+jI_L(a_0,b_0)=0\Rightarrow
\begin{cases}
R_L(a_0,b_0)=-1,\\
I_L(a_0,b_0)=0\Rightarrow b_0=\mathcal{P}(a_0).
\end{cases}$}
\end{equation}  
If the numerator does not have a root at $a_0+jb_0$, then  
\begin{equation}\label{E-50-51-2}
\resizebox{\columnwidth}{!}{$
	\begin{array}{*{35}{c}}
	\beta\left(1+R_L(a_0,b_0)+jI_L(a_0,b_0)\right)+\rho^\prime\left(R_{C_R}(a_0,b_0)+I_{C_R}(a_0,b_0)\right) \neq 0\\ \xRightarrow{(\ref{E-50-51-1})}
	\beta\neq\rho^\prime R_{C_R}(a_0,b_0)\ \lor \ I_{C_R}(a_0,b_0)\neq 0.
\end{array}$}
\end{equation}
For GFORE, by (\ref{E-50-51-2}), this yields
\begin{equation}\label{E-50-51-3}
\beta\neq\dfrac{\rho^\prime\omega_r}{a_0+\omega_r}\ \lor \ b_0\neq 0,
\end{equation} 
and for PCI
\begin{equation}\label{E-50-51-4}
\beta\neq\dfrac{\rho^\prime(a_0+\omega_r)}{a_0}\ \lor \ b_0\neq 0.
\end{equation}
Therefore, by Step 1, (\ref{E-50-51-3}) and (\ref{E-50-51-4}), it is possible to find a pair $(\beta,\rho^\prime)$ such that $H(s)$ does not have any pole-zero cancellation.   
\newline Step 5: By Steps 1-4 and Hypothesis (I), we concluded that $H(s)$ is SPR, and $(A,C_0)$ and $(A,B_0)$ are observable and controllable, respectively. Hence, according to the $H_\beta$ condition \cite{beker2004fundamental,guo2015analysis,banos2011reset}, the system is quadratically stable. To complete the proof we have to show that the system is UBIBS stable. In \cite{beker2004fundamental}, it has been shown that, for GFORE and PCI , when $\gamma=0$ and the $H_\beta$ condition is satisfied, the system is UBIBS. Here, the part of that proof related to $\gamma$ is modified, while the remaining parts of the proof are the same. $e(t_i)=0$ if $t_i$ is a reset instants. Thus,
\begin{equation}\label{E-51}
\begin{array}{*{35}{c}}
\dfrac{x_r(t_i)}{dt_i}=A_r\left(e^{A_r(t_i-t_{i-1})}x_r(t_{i-1})+\int_{t_{i-1}}^{t_i}e^{A_r(t_i-\tau)}B_re(\tau)d\tau\right)=\\
A_rx_r(t_i)\Rightarrow\abs{\dfrac{x_r(t_i)}{dt_i}}=\abs{A_rx_r(t_i)}.
\end{array}
\end{equation}    
Thus, since $|x_r(t_i)|$ is bounded \cite{beker2004fundamental}, $\abs{\dfrac{x_r(t_i)}{dt_i}}$ is bounded. Because $|x_r(t_{i}^+)|\leq|x_r(t_{i})|$, and $|x_r(t_{i})|$ and $\abs{\dfrac{x_r(t_i)}{dt_i}}$ are bounded,
\begin{equation}\label{E-52}
\exists\ K_1>0,\ \alpha>0\ |\ \forall\ t_i:\ |x_r(t_{i})|<K_1\left(1-e^{\alpha(t_i-t_{i-1})}\right).
\end{equation}  
The rest of the proof remains the same, thus, we have proved that the reset system (\ref{E-23}) is UBIBS. 
\end{proof}
\begin{remark}\label{RR}
 Since this frequency-domain theorem is based on the $H_\beta$ condition, if one of the conditions (I) and (II) is not satisfied, then the system is not quadratically stable.
 \end{remark}	
\section{an Illustrative Example}\label{sec:4}
In this section an example is used to show how Theorem \ref{T1} can be used to study stability of reset control systems. For this purpose, the stability of a precision positioning system \cite{saikumar2019constant} controlled by a reset controller is considered. In this system (Fig. \ref{F-41}), three actuators are angularly spaced to actuate 3 masses (indicated by B1, B2, and B3) which are constrained by parallel flexures and connected to the central mass D through leaf flexures. Only one of the actuators (A1) is considered and used for controlling the position of mass B1 attached to the same actuator which results in a SISO system. This positioning stage with its amplifier is well modelled by the second order mass-spring-damper system \cite{saikumar2019constant} as following.
\begin{equation}\label{E-70}
G(s)=\dfrac{1.429\times10^8}{175.9s^2+7738s+1.361\times10^6}
\end{equation}
\begin{figure}[!t]
	\centering
	\includegraphics[width=0.5\columnwidth]{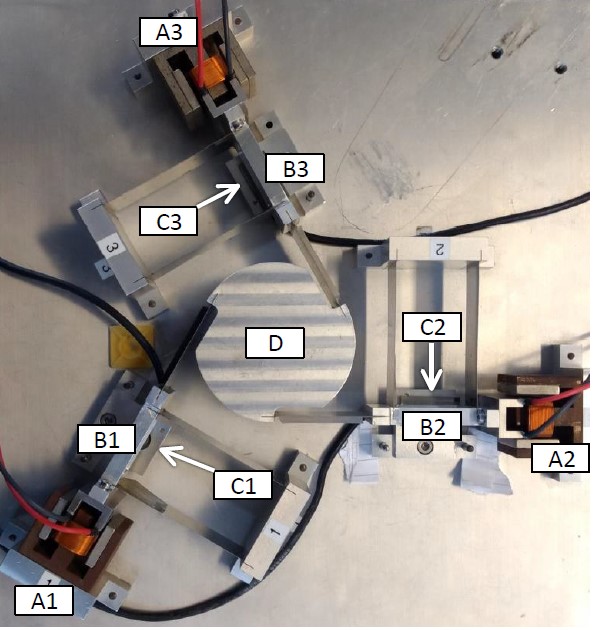}
	\caption{Spider stage}
\label{F-41}
\end{figure}	
In \cite{saikumar2019constant}, a non-linear phase compensator, which is termed Constant in gain Lead in phase (CgLp) (for more details see \cite{palanikumar2018no,chen2019development,saikumar2019constant}), has been used to improve the performance of the precision positioning stage. CgLp compensators, consisting of a lead filter and a GFORE, have been utilized along with a PID controller to give the overall controller    
\begin{equation}\label{E-71}
\resizebox{\columnwidth}{!}{$
C(s) = K_{p}\underbrace{\overbrace{\left(\cancelto{\gamma}{\dfrac{1}{\frac{ds}{\omega_c}+1}}\right)}^{\mathrm{GFORE}}\overbrace{\left(\dfrac{\frac{s}{\omega_c}+1}{\frac{s}{10\omega_c}+1}\right)}^{\mathrm{Lead}}}_{\mathrm{Reset\ Compensator}}\underbrace{\overbrace{\left(1+\frac{\omega_c}{10s}\right)}^{\mathrm{PI}}\overbrace{\left(\dfrac{\frac{gs}{\omega_c}+1}{\frac{s}{g\omega_c}+1}\right)}^{\mathrm{Lead}}\overbrace{\left(\dfrac{1}{\frac{s}{10\omega_c}+1}\right)}^{\mathrm{Low-Pass}}}_{\mathrm{PID}}.$}
\end{equation}
in which $\omega_c$ is the cross-over frequency and $K_{p},\ \gamma,\ d,\ \text{and }g$ are tuning parameters. In \cite{saikumar2019constant}, five controllers with different values of $K_{p},\ \gamma,\ d,\ \text{and }g$ (see Table \ref{Ta1}) have been designed to provide $45^\circ$ phase margin at $\omega_c=200\pi$(rad/s). This results in
\begin{table}
	\caption{Tuning parameters of controller (\ref{E-71}) \cite{saikumar2019constant}}
	\label{Ta1}
	\begin{center}
			\begin{tabular}{|c||c||c||c||c||c|}
				\hline
				\textbf{\begin{tabular}[c]{@{}c@{}}Tuning \\ Parameters\end{tabular}} & $C_1$ & $C_2$ & $C_3$ & $C_4$ & $C_5$ \\ \hline
				$K_{p}$                                                             & 0.070 & 0.163 & 0.201 & 0.197 & 0.183 \\ \hline
				$\gamma$                                                            & 0     & 0.2   & 0.4   & 0.6   & 0.8   \\ \hline
				$d$                                                                 & 1.44  & 1.23  & 1.11  & 1.04  & 1.01  \\ \hline
				$g$                                                                 & 1.98  & 2.12  & 2.27  & 2.43  & 2.63  \\ \hline
			\end{tabular}
	\end{center}
\end{table}
\begin{equation}\label{E-72}
\resizebox{\columnwidth}{!}{$
	\mathcal{L}_i(s)=\displaystyle\left(\dfrac{K_p(\dfrac{s}{200\pi}+1)(10s+200\pi)(\dfrac{gs}{200\pi}+1)1.429\times10^8}{s(\dfrac{s}{200\pi}+1)(\dfrac{s}{200\pi g}+1)(\dfrac{s}{2000\pi}+1)(175.9s^2+7738s+1.361\times10^6)}\right)
,$}
\end{equation}
\begin{equation}\label{E-73}
C_{R_i}(s)=\left(\dfrac{1}{\frac{ds}{200\pi}+1}\right),
\end{equation}
\begin{equation}\label{E-74}
L_i(s)=C_R\mathcal{L}.
\end{equation}
As the reset element used in these controllers is a GFORE and $\mathcal{L}_i$ have a pole at the origin, we use Definition \ref{D2} to assess stability. The properties of $\mathcal{N}_\chi(\omega)$ and $\mathcal{N}_\Upsilon(\omega)$ for these controllers are listed in Table \ref{Ta2}. On the basis of this table all of these reset control systems are of Type I.   
\begin{table*}
	\caption{Properties of $\vec{\mathcal{N}}(\omega)$ for the five considered reset control systems}
	\label{Ta2}
	\begin{center}
	\begin{tabular}{|c||c||c||c||c||c|}
		\hline
		Systems                                                                                & $L_1$                           & $L_2$                             & $L_3$                           & $L_4$                           & $L_5$                          \\ \hline
		\begin{tabular}[c]{@{}c@{}}$\mathcal{L}$ has a\\ pole at origin\end{tabular}             & Yes                             & Yes                               & Yes                             & Yes                             & Yes                            \\ \hline
		$\mathcal{M}$                                                                      & 279.2-6945.0                    & 495.7-7090.7                      & 630.0-7225.6                    & 686.8-7354.4                    & 718.3-7488.7                   \\ \hline
		$\mathcal{Q}$                                                                  & 80.9-256.3                      & 80.7-370.2                        & 81.2-398.9                      & 81.8-388.1                      & 82.6-368.0                     \\ \hline
		Sign($\mathcal{N}_\Upsilon(\omega\in\mathcal{M})$)                                           & +                               & +                                 & +                               & +                               & +                              \\ \hline
		Sign($\mathcal{N}_\chi(\omega\in\mathcal{Q})$)                                          & +                               & +                                 & +                               & +                               & +                              \\ \hline
		$\mathcal{I}_3$                                                                        & $\varnothing$                   & $\varnothing$                     & $\varnothing$                   & $\varnothing$                   & $\varnothing$                  \\ \hline
		$\delta_1<\psi_1$                                                                      & $0.11<0.44$                     & $0.12<0.45$                       & $0.14<0.47$                     & $0.18<0.61$                     & $0.34<1.42$                    \\ \hline
		Type              &   (I)   &  (I)   & (I)   &  (I)   &     (I)
		\\ \hline
	\end{tabular}	
	\end{center}
\end{table*}
To provide a better insight, the angels $\phase{\vv{\mathcal{N}}(\omega)}$ for these reset systems are plotted in Fig. \ref{F-42}. As demonstrated by the figure, for all of these systems $\theta_1\in(-\dfrac{\pi}{2},\pi)$, $\theta_2\in(-\dfrac{\pi}{2},\pi)$  and $\theta_2-\theta_{1}<\pi$, therefore, the condition in Remark \ref{Rs1} holds.
 \begin{figure}[!t]
 	\centering
 	\resizebox{0.81\columnwidth}{!}{
 		\begin{tikzpicture}
 		\node[anchor=south west,inner sep=0] at (0,0) {\includegraphics[width=\columnwidth]{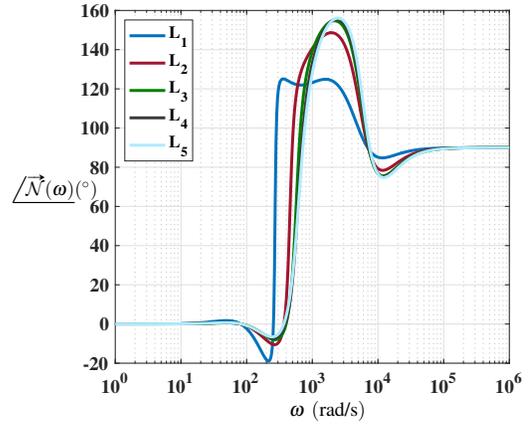}};
 		\draw (5,.2) node [scale=1.04]  {$\boldmath{\omega\ (\text{rad/s})}$};
 		\draw (0,4.3) node [scale=1.01]  {$\boldmath{\phase{\vv{\mathcal{N}}(\omega)}(^\circ)}$};
 		\end{tikzpicture}}
 		\caption{$\phase{\vec{\mathcal{N}}(\omega)}$ for the five considered reset control systems}
 		\label{F-42}
 	\end{figure}
Furthermore, the base linear systems of these controllers are stable and do not have any pole-zero cancellation in the open-loop transfer functions. Hence, by Theorem \ref{T1}, all of these controllers give UBIBS stable reset control systems.
  
In order to verify the results, the $H_\beta$ parameters for each reset system are listed in Table \ref{Ta3}. As demonstrated by the table, the $H_\beta$ condition is satisfied which is consistent with our conclusion. The step responses of the reset control systems are plotted in Fig. \ref{F-43} which demonstrates the performances of these controllers.
\begin{table}
	\caption{Pairs ($\rho^\prime,\beta$) for the five considered reset control systems}
	\label{Ta3}
	\begin{center}
		\begin{tabular}{|c||c|}
			\hline
			Systems & Equivalent $H_\beta$ ($\beta>0$)  \\ \hline
			$L_1$   & $2.24<\dfrac{\rho^\prime}{\beta}<8.77$   \\ \hline
			$L_2$   & $2.19<\dfrac{\rho^\prime}{\beta}<8.7.94$ \\ \hline
			$L_3$   & $2.12<\dfrac{\rho^\prime}{\beta}<6.85$   \\ \hline
			$L_4$   & $1.63<\dfrac{\rho^\prime}{\beta}<5.36$   \\ \hline
			$L_5$   & $0.7<\dfrac{\rho^\prime}{\beta}<2.91$    \\ \hline
	\end{tabular}	
\end{center}
\end{table}			

In summary, as shown by Table \ref{Ta2} and Fig. \ref{F-42}, the proposed results allow us determining stability of these reset control systems without computing values for the pair $(\rho,\beta)$. 
\begin{figure}[!t]
	\centering
	\includegraphics[width=0.75\columnwidth]{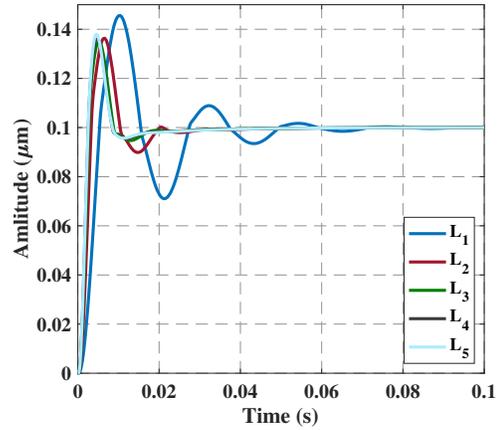}
	\caption{Step responses of the five considered reset control systems}
	\label{F-43}
\end{figure}
\section{Conclusion}\label{sec:5}
In this paper a novel frequency-domain method for determining stability properties of reset control systems has been proposed. This method is based on the $H_\beta$ condition and it can assess stability of reset control systems using the frequency response of their base linear open-loop transfer function. Consequently, this method does not need an accurate parametric model of the system and solving LMIs. The effectiveness of the proposed method has been illustrated by one practical example. This method may increase usage of reset controllers in high-precision industry to improve performances of control systems.


\section{ACKNOWLEDGMENTS}
This work has been partially supported by NWO through OTP TTW project $\#$16335, by the Erasmus institution, by the European Union's Horizon 2020 Research and Innovation Programme under grant agreement No 739551 (KIOS CoE), and by the Italian Ministry for Research in the framework of the 2017 Program for Research Projects of National Interest (PRIN), Grant no. 2017YKXYXJ.  
\bibliographystyle{IEEEtran}
\bibliography{phd}
\end{document}